\documentclass[letterpaper, 10 pt, conference]{ieeeconf}
\IEEEoverridecommandlockouts      
\author{Shuang Gao and Peter E. Caines
\thanks{*This work is supported by NSERC (Canada) and the U.S. ARL grant W911NF1910110.}
\thanks{Shuang Gao and Peter E. Caines are with the Department of Electrical and Computer Engineering, McGill University,
  Montreal, QC, Canada \hspace{1cm}
        {\tt\small    $\{$sgao,peterc$\}$@cim.mcgill.ca}. }%
}

\usepackage[noadjust]{cite}

\usepackage{graphics}
\usepackage{epstopdf}

\usepackage{amsmath}
  \usepackage[font=footnotesize]{subfig}

\usepackage{enumerate}
\usepackage{color}

\usepackage{comment}
\includecomment{comment}
\specialcomment{notes}{\begingroup\sffamily\color{blue}}{\endgroup}

\newcommand{\FA}{\mathbf{A}}   
\newcommand{\FB}{\mathbf{B}}   
\newcommand{\Fu}{\mathbf{u}}	   
\newcommand{\Fx}{\mathbf{x}}	   

\newcommand{\FAW}{\mathbf{A}}	   

\newcommand{\FU}{\mathbf{U}}	   

\newcommand{\Ff}{\mathbf{f}}	   

\newcommand{\Fz}{\mathbf{z}}	   
\newcommand{\Fp}{\mathbf{p}}

\newcommand{\SA}{\mathbf{A^{[N]}}}
\newcommand{\SB}{\mathbf{B^{[N]}}}

\newcommand{\Sxt}{\mathbf{x^{[N]}_t}}

\newcommand{\DSxt}{\mathbf{\dot{x}^{[N]}_t}}

\newcommand{\Sut}{\mathbf{u^{[N]}_t}}

\newcommand{\BI}{\mathbb{I}}	 
\newcommand{\BW}{\mathbb{W}}   
\newcommand{\BT}{\mathbb{T}}

\newcommand{\BA}{\mathbb{A}}
\newcommand{\BB}{\mathbb{B}}
\newcommand{\BR}{\mathds{R}}  

\newcommand{\BZ}{\mathds{Z}_{+}}

\newcommand{\Chi}{\mathds{1}}

\newcommand{\ESZ}{\mathcal{W}_0}
\newcommand{\ESO}{\mathcal{W}_1}

\newcommand{\GSZT}{\ESZ}
\newcommand{\GSOT}{\ESO}

\newcommand*\TRANS{{\mathpalette\doTRANS\empty}}
\makeatletter
\newcommand*\doTRANS[2]{\raisebox{\depth}{$\m@th#1\intercal$}}
\makeatother
\usepackage{url}
\usepackage{amssymb,mathtools}

\usepackage[standard,thmmarks,amsmath]{ntheorem}
\usepackage{times}
\usepackage{dsfont}
\usepackage{tikz}
\usepackage{url}

\begin{document}
%
\title{Spectral Representations of Graphons in Very Large Network Systems Control}
%
%
%
\maketitle 

\begin{abstract}
Graphon-based control has recently been proposed and developed  to solve control problems for dynamical systems on networks which are very large or growing without bound (see Gao and Caines, CDC 2017, CDC 2018). In this paper, spectral representations, eigenfunctions and approximations of graphons,  and their applications to graphon-based control are studied. First, spectral properties of graphons are presented and then approximations based on Fourier approximated eigenfunctions are analyzed. Within this framework, two classes of graphons with simple spectral representations are given. Applications to graphon-based control analysis are next presented; in particular, the controllability of systems distributed over very large networks is expressed in terms of the properties of the corresponding graphon dynamical systems. Moreover, spectral analysis based upon real-world network data is presented, which demonstrates that low-dimensional spectral approximations of networks are possible. Finally, 
an initial, exploratory investigation of the utility of the spectral analysis methodology in graphon systems control to study the control of epidemic spread is presented. 

\end{abstract}

\section{Introduction }

Graphon theory has been developed to study large networks and graph limits \cite{borgs2008convergent, borgs2012convergent,lovasz2012large}. Recently it has been applied to study dynamical systems such as the heat equation \cite{medvedev2014nonlinear}, coupled oscillators \cite{chiba2019mean} and power network synchronization \cite{Kuehn2018PowerNetsGraphon}, to analyze network centrality \cite{avella2018centrality}, to investigate static and dynamic games \cite{parise2018graphon, PeterMinyiCDC18GMFG}, and to control large networks of dynamical systems \cite{ShuangPeterCDC17, ShuangPeterCDC18, ShuangPeterTAC18,shuangPhDthesis2018,PeterMinyiCDC18GMFG}. Graphon theory provides a theoretical tool for the study of arbitrarily large and, in the limit, infinite network systems, and thus enables low-complexity approximate solutions to control problems on such systems \cite{ShuangPeterCDC17, ShuangPeterCDC18, ShuangPeterTAC18,shuangPhDthesis2018,PeterMinyiCDC18GMFG}.

Among these applications, graphon spectral properties are very significant \cite{ShuangPeterCDC19W2}.  In fact, the spectral analysis of large-scale systems has been studied since the late 1960s \cite{aoki1968control} and it plays a key role to the low-complexity control synthesis of such systems \cite{swigart2014optimal,callier1992lq,aoki1968control}. 
 This leads us to the study of spectral representations and approximations of graphons in this paper.

 As operators, graphons are Hilbert-Schmidt integral operators and hence are compact. Moreover, the symmetry property of graphons ensures that the Spectral Theorem \cite{sauvigny2012partial} applies to graphon operators.  Topics on Hilbert-Schmidt integral operators and self-adjoint compact operators have been extensively studied in the literature (see e.g. \cite{rudin1991functional,mercer1909,sauvigny2012partial}), and the spectral properties of graphons, for instance,  are investigated in  \cite{szegedy2011limits}.

This paper studies the control and analysis of graphon systems  and  their associated networks via the exploitation of their spectral properties. 
The contributions of this paper include 1) the presentation of the spectral representations of two types of graphons, 2) an analysis of the exact controllability of a class of graphon dynamical systems based on spectral decompositions, 3) the study of network spectral properties based on real-world network data which demonstrates that low-dimensional spectral approximations of networks are possible, and 4) the initial, exploratory investigation of the utility of the spectral analysis methodology in graphon systems control to study the controlled epidemic spread process

 \subsection*{Notation}
 $\BR$, $\BR_+$ and $\BZ$ represent respectively the sets of all real numbers, all positive real numbers and all positive integers.  $\langle\cdot ,  \cdot \rangle$ and $\|\cdot \|$ denote respectively inner product and norm. Bold face letters (e.g. $\FA$, $\FB$, $\Fu$) are used to represent graphons and functions.  Blackboard bold letters (e.g. $\BI$, $\BA$, $\BB$, $\BW$) are used to denote linear operators which are not necessarily compact; in particular, $\BI$ denotes the identity operator. 

\section{Graphs and Graphons}

 Network structures can be modeled as graphs. A graph $G=(V,E)$ is specified by an node set $V$ and an edge set $E \subset V \times V$. It has a representation by the corresponding adjacency matrix $A=[a_{ij}]$ where the element $a_{ij}$ is one when there is an edge from node $i$ to node $j$, and zero otherwise. For a weighted graph, the elements of its adjacency matrix represent the corresponding edge weights. Furthermore, if one embeds the adjacency matrix on $[0,1]^2$ as a pixel picture where each pixel has a side length $\frac1{|V|}$ with $|V|$ representing the cardinality of $V$, then it gives a function $\FAW:[0,1]^2\rightarrow [0,1]$.

 Formally, graphons are defined as bounded {symmetric} Lebesgue measurable functions $\FAW: [0,1]^2 \rightarrow [0,1]$, which can be interpreted as weighted graphs on the node set $[0,1]$. 
A meaningful convergence with respect to the \emph{cut metric} (see e.g. \cite{lovasz2012large}) is defined for sequences of dense and finite graphs.  Graphons are then the limit objects of converging graph sequences.
Moreover, graphons can be used as generative models for exchangeable random graphs \cite{orbanz2015bayesian}. These properties make graphons suitable to model extremely large-scale networks. 
We note that in some papers, for instance \cite{borgs2014lp}, the word ``graphon" refers to symmetric, integrable functions from $[0,1]^2$ to $\BR$.  
In this paper, unless stated otherwise, the term ``graphon'' refers to a bounded symmetric Lebesgue measurable function  ${\FAW_1}: [0,1]^2\rightarrow [-1,1]$, so as to include networks with possibly negative weights, and  $\GSOT$ denotes the set of all graphons.  Let $\GSZT$ represent the set of all graphons satisfying ${\FAW_0}: [0,1]^2\rightarrow [0,1]$.  $\GSZT$ and $\GSOT$ are both compact under the cut metric (after identifying points of zero distance) \cite{lovasz2012large}.

\section{Graphon Operators}

\subsection{Graphon Operators}

Let $L^2(\Omega)$ be the standard Lebesgue space defined on $\Omega$ endowed with  the $2$-norm $\|x\|_2 = ({\int_\Omega x(\alpha)^2 d\alpha})^{\frac{1}{2}}$. $L^2[0,1]$ and $L^2[0,1]^2$ are defined by specializing $\Omega$ to be $[0,1]$ and $[0,1]^2$, respectively.

\begin{definition}
	A linear map $\BT: L^2[0,1]\rightarrow L^2[0,1]$ is said to be \emph{compact} if $\BT$ maps the open unit ball in $L^2[0,1]$ to a set in $L^2[0,1]$ that has compact closure.

\end{definition}
\begin{definition}[\cite{lovasz2012large}]
	A \emph{graphon operator} $\BT_\FA: L^2[0,1] \rightarrow L^2[0,1]$ is defined by a graphon $\FAW \in \GSOT$ as follows:
\begin{equation}\label{equ: graphon operation on functions}
	[\BT_\FA \Ff](x) = \int_0^1 \FAW(x,y) \Ff(y) dy, \quad \Ff(\cdot) \in L^2[0,1].
\end{equation}
\end{definition}
Clearly, the operator $\BT_\FA$ is Hermitian (or self-adjoint), since for any $x,y$ in $L^2[0,1]$, $\langle x, \BT_\FA y \rangle = \langle \BT_\FA x, y\rangle$. Moreover, the graphon operator $\BT_\FA$ is a linear operator that is bounded, (hence) continuous, and compact \cite{conway2013course}.  

For simplicity of notation we henceforth use the bold face letter (e.g. $\FAW$, $\FB$) to represent both a graphon and its corresponding graphon operator. Let $\FAW \Ff$ denote the function defined by \eqref{equ: graphon operation on functions}.

The graphon operator product is then defined by
\begin{equation} \label{equ: graphon operation on graphons}
	[{\FA \FB}](x,y)=\int_0^1\FA(x,z)\FB(z,y)dz,
\end{equation}
where ${\FA,\FB} \in \GSOT.$ 
Let $\FA \FB$ denote the graphon given by the convolution in (\ref{equ: graphon operation on graphons}).
Consequently, the power ${\FAW}^n$ of an operator ${\FAW} \in \GSOT$ is given by
\begin{equation*}
	{\FAW}^n(x,y)=\int_{[0,1]^n}{\FAW}(x,\alpha_1)\cdots {\FAW}(\alpha_{n-1},y)d\alpha_1\cdots d\alpha_{n-1}
\end{equation*}
with ${\FAW}^n \in \GSOT, n\in \BZ$. 
${\FAW^0}$  is formally defined as the identity operator $\BI$ on  functions in $L^2{[0,1]}$, and hence ${\FAW^0}$ is neither a graphon nor a compact operator.
Furthermore, $e^{\FAW}:= \sum_{k=0}^\infty \frac{1}{k!}\FAW^k$ of a graphon operator $\FAW$ is a bounded linear operator from $L^2[0,1]$ to $L^2[0,1]$. 

\subsection{Spectral Representations} 

Denote the operator norm for a linear operator $\BT$ on $L^2[0,1]$ as 
\begin{equation}\label{equ:op-norm-def}
	\|\BT \|_{\text{\textup{op}}}= \sup_{ \Ff \in L^2[0,1], \|\Ff\|_2=1}\|\BT \Ff\|_2.
\end{equation}
Following \cite{rudin1991functional}, we define \emph{kernel space}, \emph{spectrum}, \emph{eigenvalue} and \emph{eigenfunction} below. 
\begin{definition}
	Define the \emph{kernel space} (or \emph{null space})  for a linear operator $\BT$ on $L^2[0,1]$ as :
$$
	\text{Ker}(\BT) := \{x \in L^2[0,1]: \BT x=0 \}.$$
\end{definition}

\begin{definition}
	The \emph{spectrum} $\sigma(\BT)$ of a linear bounded operator $\BT$ on $L^2[0,1]$ is the set of all (complex or real) scalars  $\lambda$ such that $\BT-\lambda \BI$ is not invertible, where $\BI$ is the identity operator.  Thus $\lambda \in \sigma(\BT)$ if and only if at least one of the following two statements is true: 
\begin{enumerate}
	\item[(i)] The range of $\BT -\lambda \BI $ is not all of $L^2[0,1]$, i.e., $\BT -\lambda \BI $ is not onto. 
	\item[(ii)] $\BT -\lambda \BI$ is not one-to-one.
\end{enumerate}
If (ii) holds, $\lambda$ is said to be an \emph{eigenvalue} of $\BT$; the corresponding eigenspace is $\text{Ker}(\BT- \lambda \BI)$; each $x \in \text{Ker}(\BT- \lambda \BI)$ (except $x=0$) is an \emph{eigenfunction} of $\BT$; it satisfies the equation $\BT x= \lambda x$. 
\end{definition}

\begin{lemma}[\cite{ShuangPeterTAC18}] \label{lem:operator-norm-and-L2-norm}
For any graphon $\FAW$ or any function $\FAW$ in $L^2[0,1]^2$, $$\|\FAW\|_{\text{\textup{op}}} \leq \|\FAW\|_2,$$
where  $\|\FA\|_2 =\big(\int_0^1\int_0^1 (\FA(x,y))^2 dx dy\big)^{\frac12}$.
\end{lemma}
Hence a graphon sequence convergences under $\|\cdot\|_2$ implies it convergences under $\|\cdot\|_{\text{\textup{op}}}$.
Furthermore, the following inequalities hold \cite{janson2010graphons,parise2018graphon}:
\begin{equation}\label{equ:cut-norm-operator-norm}
  \|\FAW\|_\Box \leq \|\FAW\|_{\text{\text{\textup{op}}}} \leq \sqrt{8\|\FAW\|_\Box},
\end{equation}
where the cut norm of a graphon $\FAW \in \GSOT$ is defined as
 \begin{equation}
 \| \FAW \|_{\Box}=\sup_{M,T\subset [0,1]}|\int_{M\times T}\FAW(x,y)dxdy|.
 \end{equation}
The \emph{cut metric} between two graphons $\mathbf{V}$ and $\mathbf{W}$ is then given by
\begin{equation}
	\delta_{\Box}(\mathbf{W, V})=\inf_{\phi\in S_{[0,1]}}\|\mathbf{W}^{\phi} -\mathbf{V}\|_{\Box},
\end{equation}
where $\mathbf{W}^{\phi}(x,y)=\mathbf{W}(\phi(x),\phi(y))$ and $S_{[0,1]}$ denotes the set of measure preserving bijections from $[0,1]$ to $[0,1]$. See \cite{lovasz2012large,janson2010graphons} for more details on different norms.

\begin{proposition}
Consider the graphon operator $\FAW$ corresponding to a graphon $\FAW \in \GSOT $.
Then there is a set $\{\Ff_\ell\}$  consisting of  a countable number of orthonormal elements  in $L^2[0,1]$ such that
	 the elements $\Ff_\ell$ are eigenfunctions to the eigenvalues $\lambda_\ell\in \BR$ ordered as follows:
	$$\|\FAW\|_{\text{\textup{op}}} = |\lambda_1| \geq  |\lambda_2| \geq |\lambda_3| \geq ... \geq 0, $$
	If the set $\{\Ff_\ell\}$ is infinite, we have the asymptotic behavior
	$\lim_{\ell\rightarrow \infty} \lambda_\ell = 0. $
	 Furthermore, for any  $\Fu \in L^2[0,1]$, $\FAW \Fu$ has representations as:
	$$\FAW \Fu = \sum_{\ell=1}^{\infty}\lambda_\ell \langle \Ff_\ell, \Fu\rangle  \Ff_\ell  \text{ and }  \langle \Fu, \FAW \Fu\rangle = \sum_{\ell=1}^{\infty}\lambda_\ell|\langle \Ff_\ell, \Fu\rangle|^2.$$
\end{proposition}	

This result is a special case of \cite[Theorem 7.3]{sauvigny2012partial}, where the compact Hermitian operators there are specialized to graphon operators. 

\subsection{Convergence of Eigenvalues}
For a graphon $\FAW$, the eigenvalues form two sequences $\mu_1(\FAW)\geq \mu_2(\FAW)\geq ...\geq 0$ and $\mu^{\prime}_1(\FAW) \leq \mu^{\prime}_2(\FAW)\leq ... \leq 0 $ converging to zero, where $\mu_i{(\FAW)}$ and $\mu_i^\prime{(\FAW)}$ denote respectively the $i^{th}$ non-negative eigenvalue and the $i^{th}$ non-positive eigenvalue.
\begin{theorem}[\cite{borgs2012convergent}]
	 Let $\{\FAW_i\}_{i=1}^{\infty}$ be a sequence of uniformly bounded graphons, converging in the cut metric to a graphon $\FAW$. Then for every $i\geq 1$,
	$$\mu_i(\FAW_n) \rightarrow \mu_i(\FAW )\quad \text{and} \quad {\mu}^{\prime}_i(\FAW_n) \rightarrow {\mu}^{\prime}_i(\FAW) \quad \text{ as } n \rightarrow \infty.$$
\end{theorem}

This implies that if a sequence of graphs converges in the cut metric \cite{lovasz2012large} to a graphon limit with a simple spectral characterization by a few non-zero eigenvalues, then the sequence of graphs admits simple low-dimensional spectral approximations. Furthermore, if the graphs in the sequence are increasing in size, then the low-dimensional approximations perform better as the networks increase in size.  
This can be illustrated by the sequence of random graphs generated by the Erd\"os-R\'enyi model in Fig. \ref{fig:ER}. For general random graphs generated by dense low-rank models, reasonable low-rank approximations exist \cite{chung2011spectra}.
\begin{figure}[htb] 
\centering
	\includegraphics[width=8.5cm]{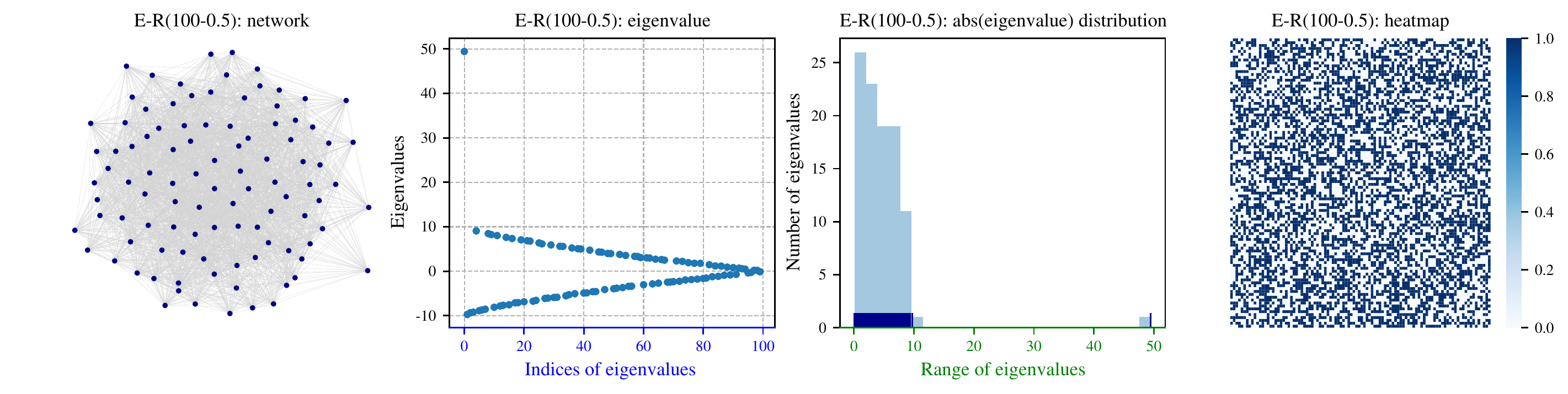}
	\caption{The eigenvalue distribution of a graph with 100 nodes in a convergent sequence of random graphs to the graphon limit $W(x,y)=0.5$. Notice that the eigenvalues accumulate at $0$ and $0.5\times100$ which correspond to the eigenvalues of the graphon $W(x,y)=0.5$. This feature makes it possible to approximate certain large random networks by low dimensional spectral characterizations. }
	\label{fig:ER}
\end{figure}

\section{Approximation of Graphons}

If the eigenvalues and the corresponding eigenfunctions of a graphon are known, 
one can approximate the graphon by a finite spectral sum.  
Consider the approximation of a graphon $\FAW$ by $\FAW_m(x,y)= \sum_{\ell=1}^m \lambda_\ell \Ff_\ell(x)\Ff_\ell(y)$.
Then the mean square error is given as follows:
\begin{equation}\label{equ:finite-spectrum-approximation-of-graphons}
	\begin{aligned}
		\|\FAW-\FAW_m\|_2 
		 = \sqrt{\|\FAW\|_2^2 - \sum_{\ell=1}^m \lambda_\ell^2 }.
	\end{aligned}
\end{equation}

Denote $p_\ell(\cdot)$ as a polynomial function and $p_\ell(e^{2\pi i \cdot})$ is used to approximate the $\ell^{th}$ eigenfunction $\Ff_\ell(\cdot)$ of $\FAW$ with $\FAW(x,y)= \sum_{\ell=1}^{\infty} \lambda_\ell \Ff_\ell(x) \Ff_\ell(y)$. Note that if the polynomial $p_\ell(\cdot)$ permits terms up to infinite order,  $p_\ell(e^{2\pi i \cdot})$ is simply the Fourier representation of $\Ff_\ell(\cdot)$. Denote the spectral sum with Fourier approximated eigenfunctions as
\begin{equation}
	\FAW_{pm}(\vartheta, \psi)= \sum_{\ell=1}^m \lambda_\ell p_\ell(e^{2\pi i  \vartheta}) p_\ell(e^{2\pi i \psi})
\end{equation}

There are two levels of approximations: (a) Fourier approximation of eigenfunctions; (b) spectral decomposition approximation of the graphon operator.
Hence the approximation error is bounded as follows:
$$
	\|\FAW-\FAW_{pm}\|_2 \leq \|\FAW -\FAW_{ m}\|_2 + \|\FAW_{ m} -\FAW_{pm}\|_2. 
$$
Moreover, the approximation error for functions of operators is given in the following. 
\begin{proposition}[\cite{shuangPhDthesis2018}]
	If there exists $c>0$ such that $\|\FAW\|_2 \leq c$ and $\|\FAW_{pm}\|_2 \leq c$, then 
	\begin{align}
		\| \FAW^n - (\FAW_{pm})^n \|_2 \leq n c^n  \|\FAW - \FAW_{pm}\|_2, \\ 
		 \|e^{\FAW}- e^{\FAW_{pm}}\|_{\text{\textup{op}}} \leq  c e^c \|\FAW-\FAW_{pm}\|_{2}.\label{equ:exponential-error}
	\end{align}
	
\end{proposition}

Consider a polynomial $p_\ell(e^{2\pi i \cdot})$ with highest order $2n$, then there exists a matrix $T_\ell\in \BR^{n\times n}$ such that

$$
	p_\ell(e^{2\pi i \cdot}) = (e_1, e_2,... e_n) T_\ell (e_1, e_2,... e_n)^T
$$
with $e_k= e^{2\pi k i  \cdot }$.
Since Fourier basis forms a complete basis for $L^2[0,1]$ (see e.g. \cite{saxe2013beginning}), any function $\Ff \in L^2[0,1]$ can be approximated by finite Fourier series and hence by polynomial functions of $e^{2\pi i \cdot}$. 

\section{Graphons with Simple Spectral Representations}

\subsection{Sinusoidal Graphons}
A \emph{sinusoidal graphon} is defined as any graphon that can be represented by 
\begin{equation}\label{equ:sinusoidal-graphon-def}
	\FAW(\varphi,\vartheta) = a_0 + \sum_{k=1}^{\infty}b_k \cos(2\pi k(\varphi-\vartheta)), \quad (\varphi,\vartheta) \in [0,1]^2.
\end{equation}
Clearly sinusoidal graphons are symmetric and diagonally constant and hence they are suitable to fit Toeplitz matrices \cite{gray2006toeplitz}. 

Sinusoidal graphons have simple spectral characterizations. 
The eigenfunctions of the graphon in (\ref{equ:sinusoidal-graphon-def}) are  $L^2[0,1]$ functions as follows:
\begin{equation*}
	1, \{\sqrt{2}\cos2\pi k (\cdot):  k \in \BZ \}, \{\sqrt{2}\sin2\pi k(\cdot): k \in \BZ \}.
\end{equation*}
The corresponding eigenvalues are:
	$a_0, \{\frac{b_k}2: k \in \BZ\}, \{\frac{b_k}2: k \in \BZ \}.$
Moreover, the eigenfunctions  
form a complete orthonormal basis for $L^2[0,1]$ (see \cite{saxe2013beginning}).

\subsubsection{Operation on $L^2[0,1]$ functions}
Consider a function $ \Fx \in L^2[0,1]$ represented by Fourier series as  
\begin{equation*}
	\Fx({\vartheta}) = K_0 + \sum_{k=1}^{\infty} (\alpha_k \sin(\vartheta)+ \beta_k \cos(\vartheta)), \quad \vartheta \in[0,1].
\end{equation*}
Then operating a sinusoidal graphon $\FAW$ on $\Fx$ yields
\begin{equation*}
	[\FAW \Fx](\vartheta) = a_0K_0+ \sum_{k=1}^{\infty} \frac{b_k}{2}(\alpha_k \sin(\vartheta)+ \beta_k\cos(\vartheta)), ~ \vartheta \in[0,1].
\end{equation*}
Similarly, for $m\in \BZ$ and $\vartheta \in[0,1],$
\begin{equation*}
	\begin{aligned}
		& [\FAW^m \Fx](\vartheta) = a_0^m K_0+ \sum_{k=1}^{\infty} \left(\frac{b_k}{2}\right)^m(\alpha_k \sin(\vartheta)+ \beta_k\cos(\vartheta)),\\
		& [e^{\FAW} \Fx](\vartheta) = e^{a_0}K_0+ \sum_{k=1}^{\infty} e^{\frac{b_k}{2}}(\alpha_k \sin(\vartheta)+ \beta_k\cos(\vartheta)).
	\end{aligned}
\end{equation*}

\subsubsection{Functions of sinusoidal graphons}
Power functions of a sinusoidal graphon $\FAW$ are given by
\begin{equation*}
	\begin{aligned}
		&\FAW^m(\varphi,\vartheta) = a_0^m + \sum_{k=1}^{\infty} (\frac{b_k}{2})^{m-1} \cdot b_k \cos(2\pi k(\varphi-\vartheta)), 
	\end{aligned}
\end{equation*}
where $(\varphi,\vartheta) \in [0,1]^2,  m\in \BZ$.
Moreover, the time-varying exponential function of a sinusoidal graphon is given by $e^{\FAW t} = \BI + {\FU}_t, ~ t\in \BR,$
where $\FU_t$ is defined as follows: for all $(\varphi,\vartheta) \in [0,1]^2$ 
\begin{equation*}
\begin{split}
		&\FU_t(\varphi, \vartheta) = \left(e^{a_0 t} -1\right) + 2\sum_{k=1}^{\infty} \left(e^{\frac{b_k }{2}t}-1\right)\cos 2\pi k(\varphi-\vartheta). \\
\end{split}
\end{equation*}
We note that, for any $t \in \BR,$  $e^{\FAW t}$ is an element of the graphon unitary operator algebra \cite{ShuangPeterCDC18}. %

\subsection{Step Function Graphons}
A graphon ${\FAW} \in \GSOT$ is a {\it step function} if there is a partition $Q=\{Q_1,...,Q_N\}$ of $[0, 1]$ into measurable sets such that ${\FAW}$ is constant on every product set $Q_i \times Q_j$. %
 A \emph{uniform partition} $P^{N}=\{P_1, P_2, ..., P_{N}\}$  of $[0,1]$ is given  by setting $P_k=[\frac{k-1}{N}, \frac{k}{N}), k\in\{1, N-1\}$ and $P_N=[\frac{N-1}N,1]$. In this paper, we focus on step functions with uniform partitions. For general step functions,  see \cite{lovasz2012large}.

The \emph{piece-wise constant function} $\mathds{S}_u$ in $L^2[0,1]$ corresponding to $u \in \BR^N$ is defined as 
\begin{equation}\label{equ:piece-wise-constant}
	\mathds{S}_u(x) =\sum_{i=1}^N {\Chi}_{_{P_i}}(x) u_i , \quad \forall x \in [0,1]
\end{equation}
where $\Chi_{_{P_i}}(\cdot)$ denotes the indicator function, that is, $\Chi_{_{P_i}}(x)=1$ if $x\in P_i$ and $\Chi_{_{P_i}}(x)=0$ if $x\notin P_i$.
Let $\mathds{S}_u\cdot \mathds{S}_v^T$ be given as follows:  
\begin{equation}
	\begin{aligned}
		[\mathds{S}_u\cdot \mathds{S}_v^T](x,y) &: =  \sum_{i=1}^N {\Chi}_{_{P_i}}(x) u_i \sum_{i=1}^N {\Chi}_{_{P_j}}(y) v_j\\
	\end{aligned}
\end{equation}
where $\mathds{S}_u$ and $ \mathds{S}_v^T$ share the same uniform partition $P^N=\{P_1,..., P_N \} $.

\begin{proposition}\label{prop:spectral-decomposition}
If the matrix $A=[a_{ij}]$ has a spectral decomposition $A=V\Lambda_d V^
\TRANS$, where $\Lambda_d=\text{diag}(\lambda_1,...,\lambda_d)$ and $V=(v_1,....,v_d)$ with $v_\ell$ representing the normalized eigenvector of $\lambda_\ell$, 
then the step function graphon  $\FAW$ defined by 
\begin{equation}\label{equ:step-function}
	\FAW(x,y) = \sum_{i=1}^{N} \sum_{j=1}^{N} \Chi_{_{P_i}}(x)\Chi_{_{P_j}}(y)a_{ij},  ~~ (x,y) \in [0,1]^2
\end{equation}
has a spectral representation given by
\begin{equation}\label{equ:step-function-spectra}
	\begin{aligned}
		\FAW(x,y) 	
				& = \sum_{\ell =1}^d \lambda_{\ell} [\mathds{S}_{v_\ell}\cdot \mathds{S}_{v_\ell}^\TRANS](x,y), ~~ (x,y)\in [0,1]^2,
	\end{aligned}
\end{equation}
where the underlying partition $\{P_1, ..., P_N\}$ of $[0,1]$ is uniform. 
\end{proposition}
\begin{proof}
	For all $(x,y) \in [0,1]^2$,
	\begin{equation*}
		\begin{aligned}
		\FAW(x,y) 
			& = \sum_{i=1}^N\sum_{j=1}^N\Chi_{_{P_i}}(x)\Chi_{_{P_i}}(y) \sum_{\ell =1}^d \lambda_{\ell} v_l(i)v_l(j) \\
				&  = \sum_{\ell =1}^d \lambda_{\ell} \sum_{i=1}^N\Chi_{_{P_i}}(x)v_\ell(i) \sum_{j=1}^N\Chi_{_{P_j}}(y)v_\ell(j)\\
				& = \sum_{\ell =1}^d \lambda_{\ell} [\mathds{S}_{v_\ell}\cdot \mathds{S}_{v_\ell}^T](x,y),
		\end{aligned}
	\end{equation*}
	where $v_\ell(i)$ denotes the $i^{th}$ element of $v_\ell$.
\end{proof}

We note that 
$\langle \mathds{S}_{v_\ell}, \mathds{S}_{v_k} \rangle  =0, \text{ if } \ell \neq k; ~ \langle \mathds{S}_{v_\ell}, \mathds{S}_{v_k} \rangle =\frac1N, \text{ if } \ell = k$ and hence the corresponding eigenvalues for $\FAW$ are given by $\{\frac1N \lambda_\ell\}_{\ell=1}^{d}$.

Piece-wise constant functions in $L^2[0,1]$ form eigenfunctions of step function graphons. Since piece-wise constant functions in $L^2[0,1]$ form a dense subset of $L^2[0,1]$ space, they can also be used to approximate eigenfunctions of general graphons.

\section{Graphon Dynamical  Systems}

Consider a group of linear dynamical 
 subsystems $\{S_i^N; 1\leq i \leq N \}$ coupled over an undirected graph $G_N$.
 The subsystem $S_i^N$ at the node $i$ of $G_N$ has  interactions with $S_j^N, 1\leq j \leq N,$ specified as below:
 \begin{equation} \label{equ:network-system}
 	\begin{aligned}
 	 &\dot{x}^i_t= \alpha_0 x_t^i+ \frac{1}{N}\sum_{j=1}^{N} {a}_{ij} x^j_t+\beta_0 u_t^i +\frac1{N} \sum_{j=1}^N {b}{_{ij}} u^j_t, \quad \\
 	& t \in [0,T],\quad \alpha_0, \beta_0 \in \BR, \quad  x^i_t, u^i_t \in \BR,
 	\end{aligned}
 \end{equation}  
  with ${A}_{N}= [{a}_{ij}]$ and $ {B}_{N} =[{b}_{ij}] \in \BR^{N\times N}$ as the symmetric adjacency matrices of $G_N$ and  of the input graph. Let $x_t=[x_1, \dots, x_N]^\TRANS$ and $u_t=[u_1, \dots, u_N]^\TRANS$.

Define the graphon step functions $\SA$ and $\SB$ that correspond to $A_N$ and $B_N$ according to \eqref{equ:step-function}, respectively. Next, define the piece-wise constant functions $\Sxt$ and $\Sut$ that correspond respectively to $x_t$ and $u_t$ according to  \eqref{equ:piece-wise-constant}.  Then the corresponding graphon dynamical system is given by 
\begin{equation}
	\begin{aligned} \label{equ:step-function-dynamical-system}
	&\DSxt= (\alpha_0 \BI+\SA) \Sxt+(\beta_0 \BI+\SB) \Sut, \quad t\in[0,T],\\
	&\alpha_0, \beta_0 \in \BR,   \quad \Sxt, \Sut \in L_{pwc}^2[0,1],\quad  \SA, \SB \in \GSOT
	\end{aligned}
\end{equation}
 where $L^2_{pwc}[0,1]$ represents the set of all piece-wise constant functions in $L^2[0,1]$.
 The trajectories of the system in \eqref{equ:network-system} correspond one-to-one to the trajectories of the system \eqref{equ:step-function-dynamical-system}. 

We formulate the infinite dimensional graphon linear system as follows: 
\begin{equation} \label{equ:infinite-system-model}
	 \begin{aligned}
		\mathbf{\dot{x}}_t=& (\alpha_0 \BI +\FA) \Fx_t + (\beta_0 \BI +\FB)\Fu_t, \quad 
		 t \in [0,T],
		\end{aligned}
\end{equation}
where $~ \alpha_0, \beta_0 \in \BR$,  $\FA, \FB \in \GSOT$, and $\Fx_t, \Fu_t \in L^2[0,1]$. ${\Fx_t}$ and ${\Fu_t}$ represent respectively the system state and the control input at time $t$. 
The space of admissible control is taken to be  $L^2 ([0, T]; L^2{[0,1]} )$, that is, the Banach space of equivalence classes
of strongly measurable mappings $\Fx:[0,T] \rightarrow  L^2{[0,1]} $ that are integrable with the norm $\|\Fx\|_{L^2 ([0, T]; L^2{[0,1]})}= (\int_0^T \int_0^1\Fx_{\tau}(\alpha)^2d\alpha d\tau )^{\frac12}$.

Evidently, $\alpha_0 \BI +\FA$ and $\beta_0 \BI +\FB$ are bounded linear operators on $L^2[0,1]$. Therefore  the system model is well defined and has a unique mild solution following \cite{bensoussan2007representation}. 
For simplicity, let $\BA = (\alpha_0\BI+\FA)$ and $\BB = (\beta_0\BI+\FB)$. Hence $\BA$ and $\BB$ lie in the graphon unitary operator algebra \cite{ShuangPeterTAC18}. Denote the graphon dynamical system in \eqref{equ:infinite-system-model} by $(\BA; \BB)$. 

The system in \eqref{equ:infinite-system-model} can represent the limit system for \eqref{equ:step-function-dynamical-system} when the underlying step function graphons convergence in the operator sense or the $L^2[0,1]^2$ sense  \cite{ShuangPeterTAC18}. 

\section{Controllability Analysis based on Spectral Representations}

A graphon $\FA$ is a compact operator, and hence it has a discrete spectrum. Its spectral decomposition is given as follows
  $
  \FA (x,y) = \sum_{\ell \in I_{\lambda}} \lambda_{\ell} \Ff_{\ell}(x) \Ff_{\ell}(y)
  $ for all $(x,y) \in [0,1]^2$.
  where $\Ff_{\ell}$ is the normalized eigenfunction corresponding to the non-zero eigenvalues $\lambda_{\ell}$ and $I_\lambda$ is the index set for non-zero eigenvalues of $\FA$, which contains a countable number of elements \cite{lovasz2012large}. 

  For infinite dimensional systems there are two notions of controllability: approximate controllability and exact controllability \cite{vidyasagar1970controllability}. 
   We only discuss exact controllability in this paper.

 \begin{definition}
		A graphon dynamical system $(\BA; \BB)$ in \eqref{equ:infinite-system-model} is \emph{exactly controllable} in  $L^2[0,1]$ over the time horizon $[0,T]$ if the system state can be driven to the origin at time $T$  from any initial state $\Fx_0\in L^2[0,1]$. 
\end{definition}
We define the \emph{controllability Gramian operator} as 
		$\BW_{T}= \int_0^{T} e^{\BA\tau}\BB\BB^\TRANS e^{\BA^\TRANS\tau}d\tau.$
If there exists $c>0$ such that, for every $\lambda \in \sigma(\BW_T)$, $|\lambda|\geq c$ holds, then the system is exactly controllable \cite{ShuangPeterTAC18} and the minimum control energy  $J =\int_0^T\|\Fu_t\|_2^2 dt $ (see. e.g. \cite{ShuangPeterCDC17}) required to drive the system from state $\Fx_0 \in L^2[0,1]$ to the origin at time $T$ is given by
$
	J(\Fx_0) = \langle e^{\BA T}\Fx_0, \BW_T^{-1}e^{\BA T}\Fx_0 \rangle .
$

For any graphon system $(\BA; \BB)$ with a compact operator $\BB$, exact controllability cannot be achieved over a finite horizon \cite{triggiani1975lack}. If  $\BB$ lies in the graphon unitary operator algebra \cite{ShuangPeterTAC18}, then exact controllability of $(\BA; \BB)$ in \eqref{equ:infinite-system-model} over $[0,T]$ implies $\beta_0 \neq 0$.

In general, it is not obvious how to find the explicit forms of the controllability Gramian operator. However, when $\FA$ and $\FB$ share the same structure, explicit forms are possible. 
\begin{proposition}
  Assume $\FA \in \GSOT$ and  $\FB=\sum_{k=1}^{d}\beta_k\FA^k$.	 Denote $\eta_\ell = \sum_{k=0}^{d}\beta_k\lambda_\ell^k$.
  Then the controllability Gramian operator for the system $(\BA, \BB)$ in \eqref{equ:infinite-system-model} is explicitly given by 
		\begin{equation} \label{equ:explicit-gramian-poly-BABB}
		\begin{aligned}
			\BW_T &= \int_0^T e^{\alpha_0 t}dt\beta_0^2\BI\\ 
			&+ \sum_{\ell \in I_{\lambda}}  \left((\eta_\ell)^2\int_0^Te^{2(\alpha_0+\lambda_\ell) t}dt- \int_0^T e^{\alpha_0 t}dt\beta_0^2 \right)  \Ff_\ell \Ff^\TRANS_\ell ;
		\end{aligned}
		\end{equation}
furthermore, if $\beta_0\neq 0$, then
the inverse of the controllability Gramian operator for $(\BA;\BB)$ in \eqref{equ:infinite-system-model} is explicitly given by
\begin{equation} \label{equ:explicit-inverse-gramian-poly-BABB}
\begin{aligned}
		\BW_T^{-1}&= \frac{1}{\int_0^T e^{\alpha_0 t}dt \beta_0^2}\BI \\
	&-  \frac{1}{\int_0^T e^{\alpha_0 t}dt\beta_0^2} \sum_{\ell \in I_{\lambda}} \frac{(\eta_\ell)^2\int_0^T e^{2\lambda_\ell t}dt-T\beta_0^2}{(\eta_\ell)^2\int_0^T e^{2\lambda_\ell t}dt} \Ff_{\ell} \Ff_\ell^{\TRANS}.
\end{aligned}
\end{equation}
  \end{proposition}
\begin{proof}
Consider any $\Fz \in L^2[0,1]$. Let $\breve \Fz = \Fz -\sum_{\ell \in I_{\lambda}} \langle \Fz, \Ff_\ell \rangle \Ff_\ell.$ Then
	\begin{equation} \label{equ:opt-inverse-BABB}
		\begin{aligned}
	&\BW_T \Fz 
	 = \int_0^T e^{\BA t} \BB \BB^\TRANS e^{\BA^\TRANS t}dt \left(\breve \Fz +\sum_{\ell \in I_{\lambda}} \langle \Fz, \Ff_\ell \rangle \Ff_\ell\right)\\
	& = \int_0^T e^{\BA t} \BB \BB^\TRANS e^{\BA^\TRANS t}dt \breve \Fz +\sum_{\ell \in I_{\lambda}} \int_0^T e^{\BA t} \BB \BB^\TRANS e^{\BA^\TRANS t} dt \langle \Fz, \Ff_\ell \rangle \Ff_\ell\\ 
	& = \int_0^T e^{2\alpha_0 t} dt\beta_0^2 \breve \Fz +\sum_{\ell \in I_{\lambda}} \int_0^T e^{2(\alpha_0+\lambda_\ell) t} \eta_\ell^2 dt \langle \Fz, \Ff_\ell \rangle \Ff_\ell\\
	& = \int_0^T e^{2\alpha_0 t} dt \beta_0^2 \Fz \\
	&+\sum_{\ell \in I_{\lambda}} \left(\eta_\ell^2\int_0^T e^{2(\alpha_0+\lambda_\ell) t}  dt - \int_0^T e^{2\alpha_0 t} dt \beta_0^2 \right)\langle \Fz, \Ff_\ell \rangle \Ff_\ell.\\
		\end{aligned}		
	\end{equation}
This yields the equivalent representation in \eqref{equ:explicit-gramian-poly-BABB}.	

If $\beta_0 \neq 0$, $\BW_T$ is invertible \cite{ShuangPeterTAC18}. Suppose $\Fu=\BW_T\Fz$. To find $\BW_T^{-1}$, we need to find the operator that maps $\Fu$ back to $\Fz$. Taking the inner product with $\Ff_\ell$ on both sides of \eqref{equ:opt-inverse-BABB} yields: 
\begin{equation} \label{equ:midle-inverse-BABB}
	\begin{aligned}
		\langle \Fu, \Ff_\ell \rangle 
		& = 	 \sum_{\ell \in I_{\lambda}} \left(\eta_\ell^2 \int_0^T e^{2(\alpha_0+\lambda_\ell) t} dt  \right)\langle \Fz, \Ff_\ell \rangle. 
\end{aligned}
\end{equation}
Then by replacing $\langle \Fz, \Ff_\ell \rangle$ in \eqref{equ:opt-inverse-BABB} based on \eqref{equ:midle-inverse-BABB}, we obtain
\begin{equation}
\begin{aligned}
	\Fz & =\frac{1}{\int_0^T e^{\alpha_0 t}dt \beta_0^2}\Fu \\
	&-  \frac{1}{\int_0^T e^{\alpha_0 t}dt\beta_0^2}\sum_{\ell \in I_{\lambda}} \frac{(\eta_\ell)^2\int_0^T e^{2\lambda_\ell t}dt-T\beta_0^2}{(\eta_\ell)^2\int_0^T e^{2\lambda_\ell t}dt} \langle \Fu ,\Ff_{\ell} \rangle \Ff_\ell^{\TRANS},
\end{aligned}
\end{equation}
and hence 
\eqref{equ:explicit-inverse-gramian-poly-BABB} holds.
\end{proof}

  If $\alpha_0=0$, then the controllability Gramian operator for the system $(\BA, \BB)$ in \eqref{equ:infinite-system-model} is given by 
		\begin{equation} \label{equ:explicit-gramian-poly}
			\BW_T = T\beta_0^2\BI + \sum_{\ell \in I_{\lambda}}  \left((\eta_\ell)^2\int_0^Te^{2\lambda_\ell t}dt- T \beta_0^2 \right)  \Ff_\ell \Ff^\TRANS_\ell ;
		\end{equation}
furthermore, if $\beta_0\neq 0$, then
the inverse of the controllability Gramian operator for $(\BA;\BB)$ in \eqref{equ:infinite-system-model} is explicitly given by
\begin{equation*} 
	\BW_T^{-1}= \frac1{T\beta_0^2} \BI - \frac1{T\beta_0^2} \sum_{\ell \in I_{\lambda}} \frac{(\eta_\ell)^2\int_0^T e^{2\lambda_\ell t}dt-T\beta_0^2}{(\eta_\ell)^2\int_0^T e^{2\lambda_\ell t}dt} \Ff_{\ell} \Ff_\ell^{\TRANS}.
\end{equation*}
These then recover the result on exact controllability in \cite{ShuangPeterTAC18}.

\begin{figure}[htb]
\centering
\vspace{0.23cm}
	\subfloat[C-elegans metabolic network where edges represent metabolic reactions between substrates \cite{jeong2000large}.]{\includegraphics[width=8.7cm]{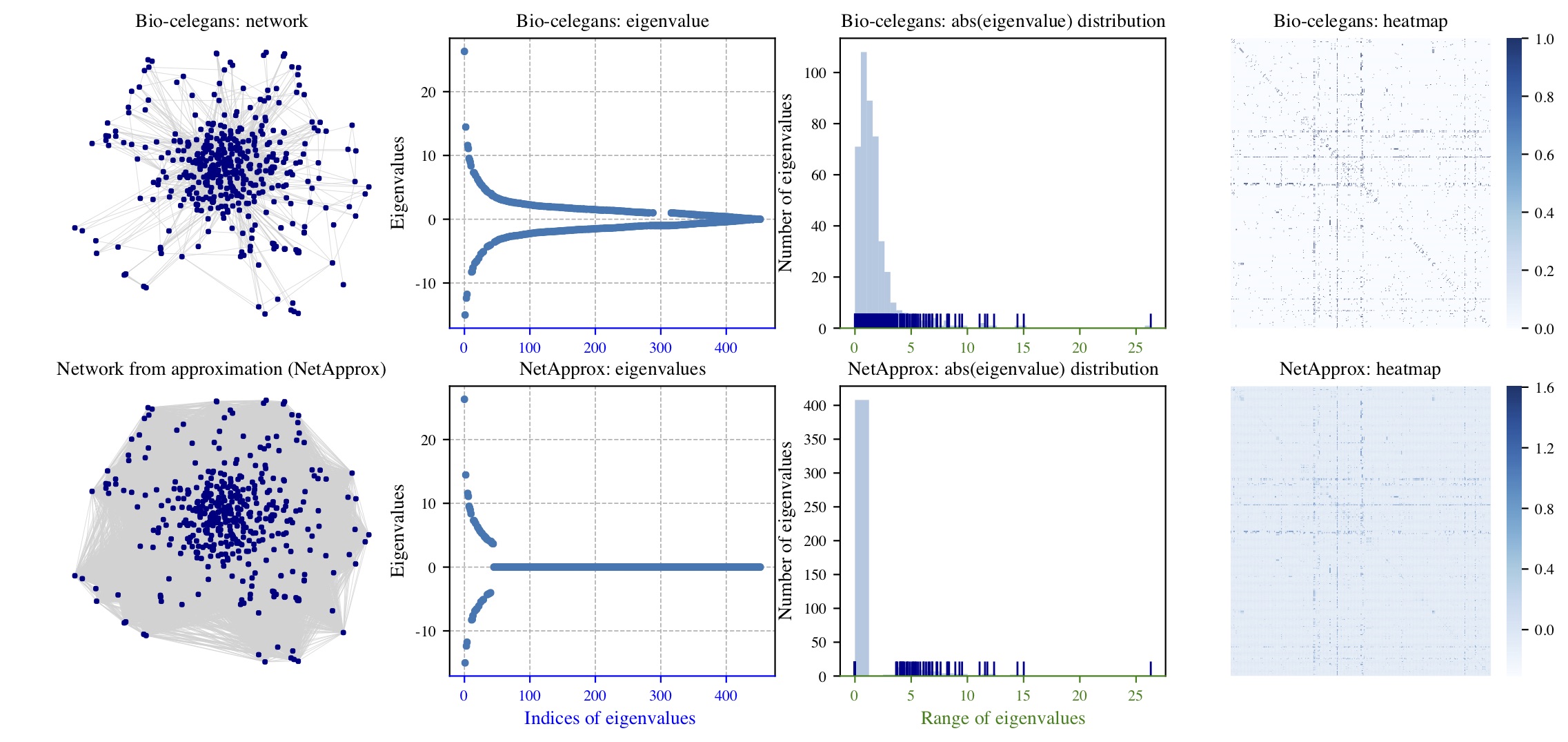}}\\
%
	\subfloat[Infectious contact network \cite{infect}.]{\includegraphics[width=8.7cm]{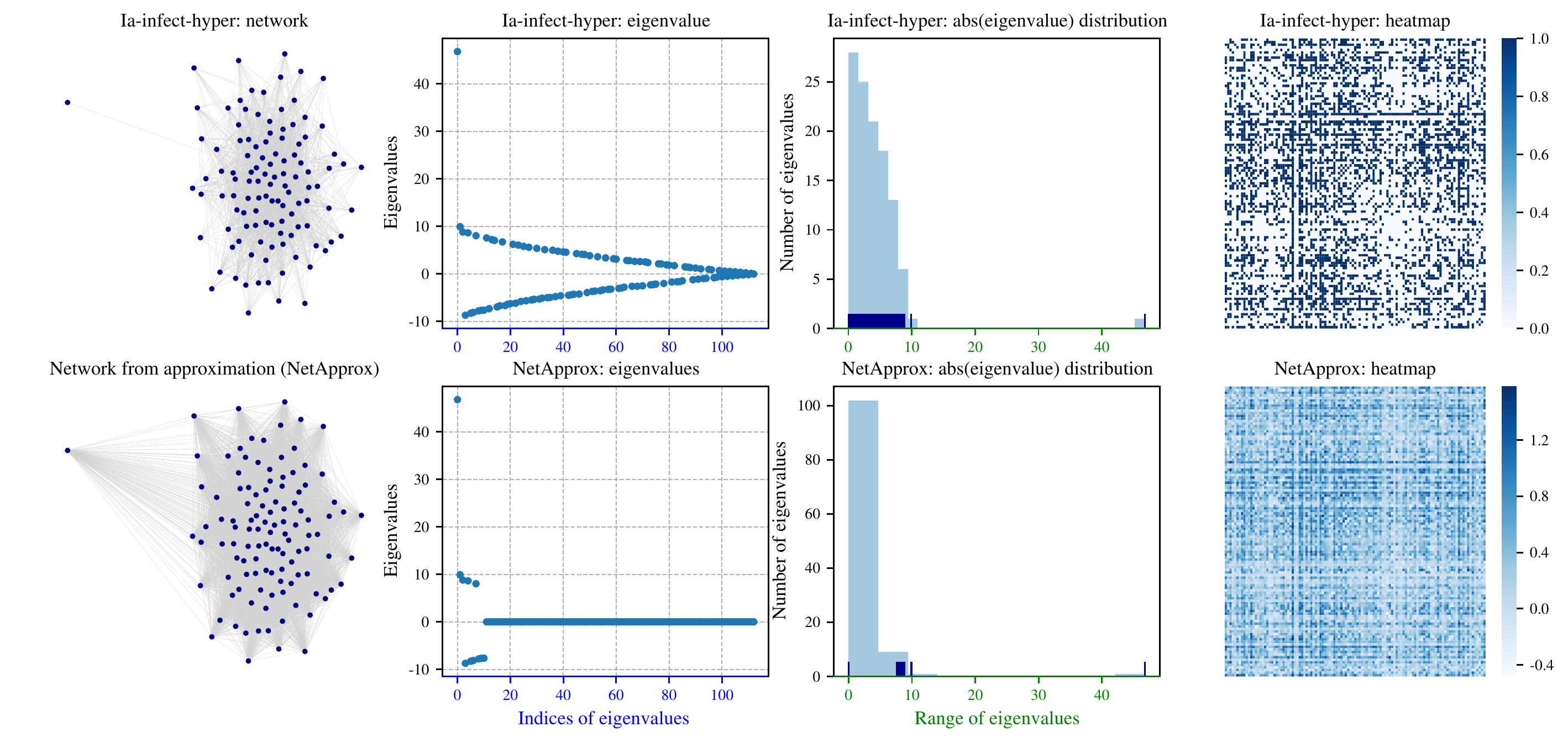}}\\	
	\subfloat[Zachary karate club network \cite{zachary1977information}.]{\includegraphics[width=8.7cm]{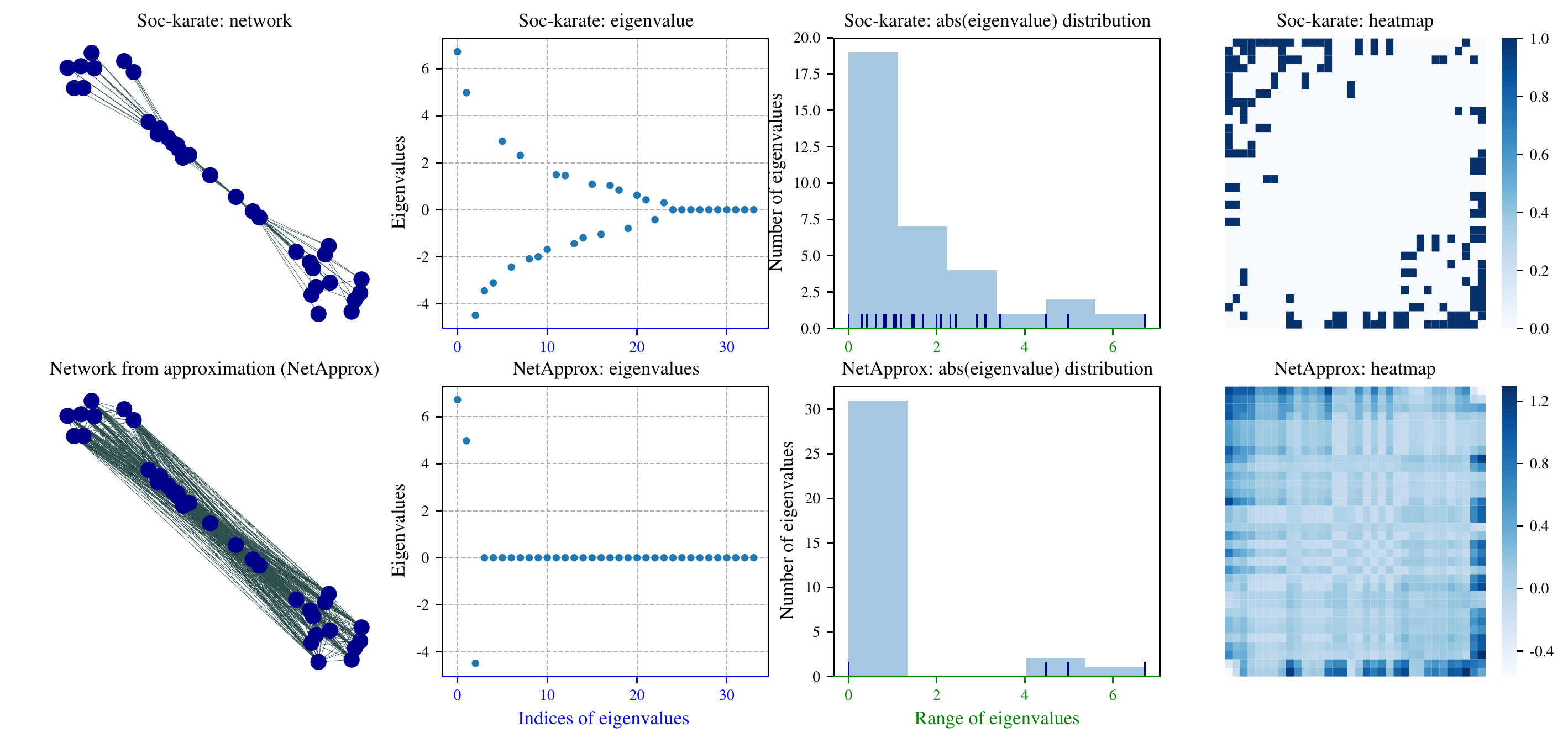}}\\
	\caption{Structures, eigenvalues, distributions of the absolute values of eigenvalues, and the corresponding step function graphons for networks and their spectral approximations by 10\% of the most significant eigendirections. 
	}
	\label{fig:spectral-approximation}
	\vspace{-0.3cm}
\end{figure}
\section{Spectral Analysis of Network Data}
Real world networks are finite in size and can be represented by step function graphons. The corresponding eigenfunctions for non-zero eigenvalues are necessarily $L^2_{pwc}[0,1]$ functions. 
 By analyzing these finite networks, we infer possible properties of the limits if such limits exist. In this section, numerical properties of the low-rank approximations to finite  networks are analyzed. 

The spectral properties of real-world network structures based on an open data set available on-line \cite{nr-aaai15} are shown in Fig. \ref{fig:spectral-approximation}. It is observed that most of the eigenvalues of many symmetric networks are distributed around zero and a few eigenvalues are large. This implies that low dimensional approximations of these networks are possible. Notice that for these networks the data set only characterizes interactions among different nodes and hence there are no self-loops in the network structures. This means that the diagonal elements of the corresponding adjacency matrices are zeros and hence the trace (i.e. sum of all eigenvalues) of each is zero. 
Note that only connection structures are captured in the data sets and the underlying dynamical systems need to be investigated in the future.  

The spectral approximation by 10\% of the most significant eigendirections is given for each network data. As shown in Fig. \ref{fig:spectral-approximation}, this approximation preserves the patterns of the corresponding graphon. In practice, the threshold for selecting the eigenvalues in this approximation depends on the tolerance of the approximation error. 
The spectral approximations of these sparsely connected graphs typically give rise to graphs with dense structures since the eigenvectors corresponding to the removed eigenvalues contain many non-zero elements.  

\section{Controlling Epidemic Processes on Networks Based on Spectral Decomposition}
Consider the process of an infectious disease spreading over contact networks where controls via vaccinations and medications are possible. Each node on the network has a state representing the infection level and each node can have a control action to receive medications or vaccinations to reduce the level of the infection. The nodes affect each other through the underlying contact network and their actions influence each other. The objective is to optimally reduce the level of the infection in the whole network with some cost constrains. 

\begin{figure*}[t]
\centering
\subfloat[Network structure and spectral properties]{\includegraphics[width=7.5cm]{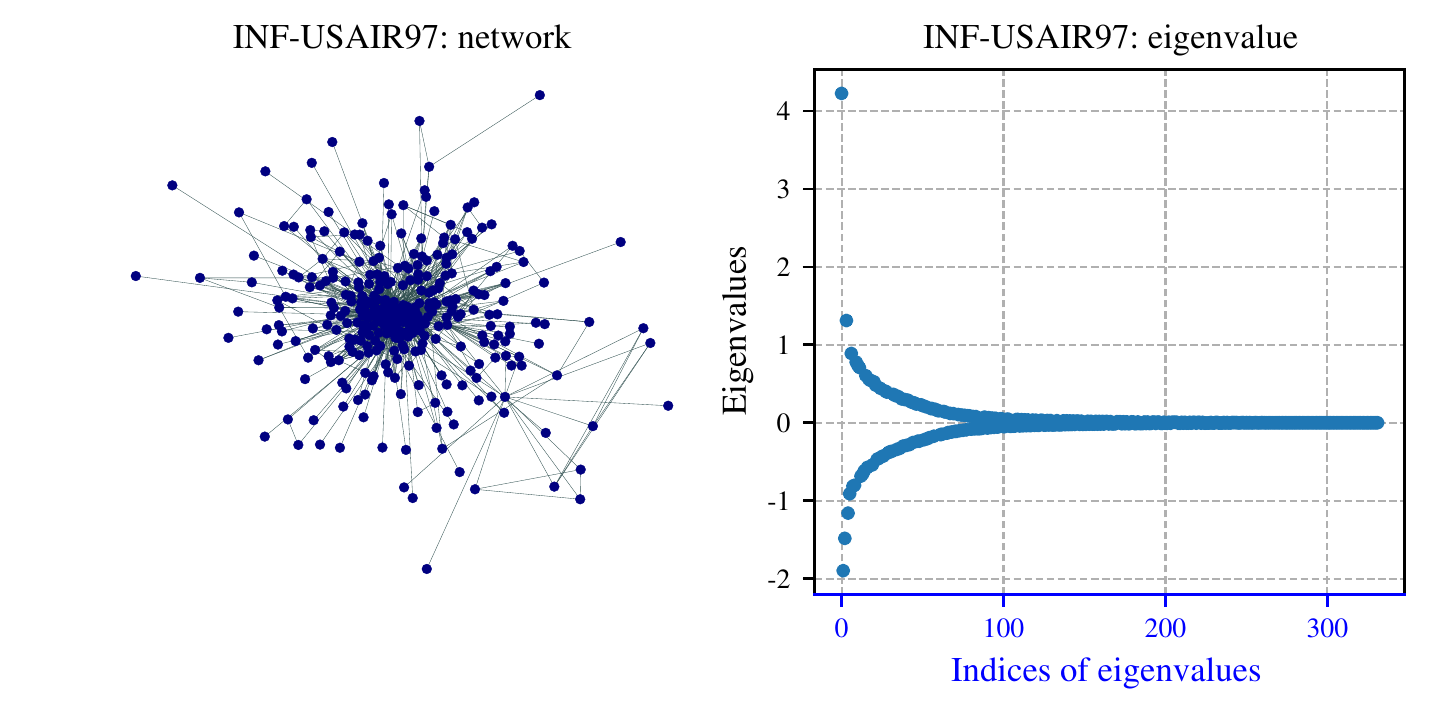}}
\subfloat[State and control trajectories for the controlled disease process]{\includegraphics[width=10cm]{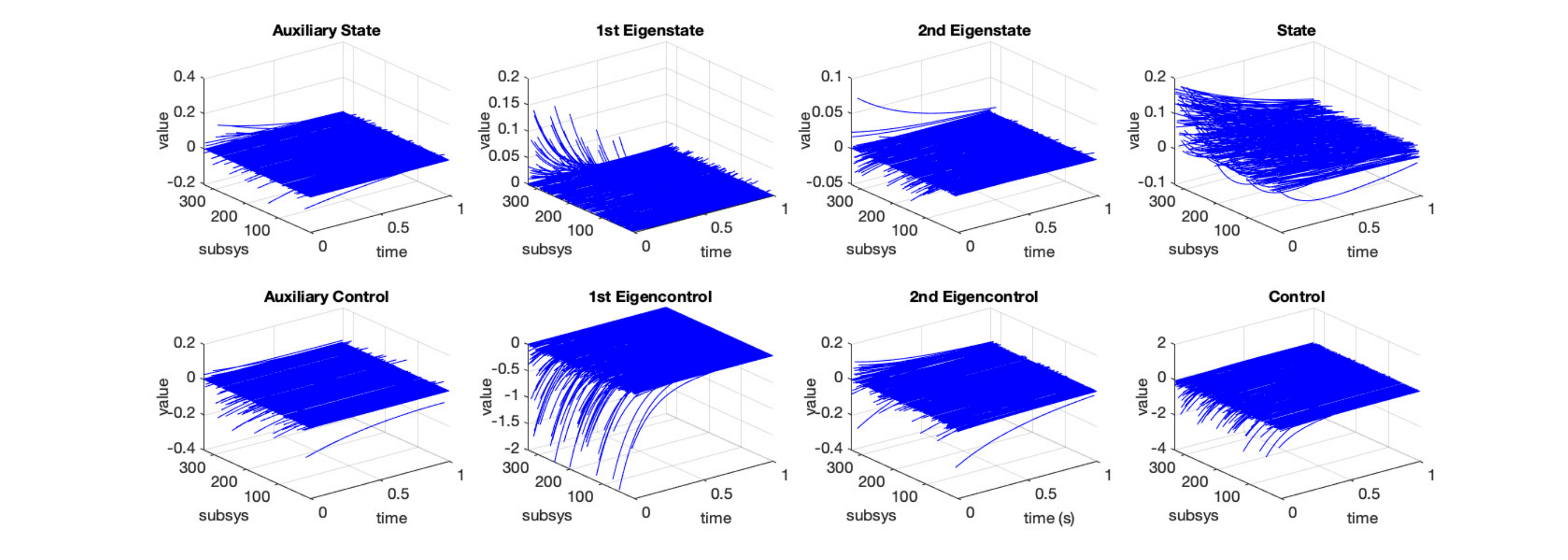}}
\caption{The simulation of the controlled disease process with couplings represented by the contact network corresponding to USAir97 \cite{USAir97}. In (b), the eigenstates and eigencontrols in only the two most significant eigendirections are shown due to space limitations.  } \label{fig:sim}
\end{figure*}

Based on the meta-population model for the epidemic process \cite{nowzari2016analysis}, the dynamics of the spread over a network is described by: 
\begin{equation}\label{eq:epedemics-dynamics}
	\dot{p}^i_t =-\alpha p^i_t +\eta \sum_{j=1}^Na_{ij}p^j_t(1-p_t^i), \quad t\in \BR_+,
\end{equation}
where $p^i_t \in [0,1]$ is interpreted as the fraction of the $i^{th}$ subpopulation that is infected, $\alpha$ is the recovering rate, $\eta$ is the infection strength and $N$ is the number of subpopulations (i.e. communities, cities). 
The origin is a global asymptotic stable equilibrium if and only if $\alpha \geq \eta\lambda_{\max}(A)$ 
\cite{nowzari2016analysis}. 
If the underlying networks grows (i.e. more nodes are connected to the networks), the limit spectral properties of the networks would be useful to estimate the $\lambda_{\max}(A) $.
It is also important to recognize the network eigendirections with significant eigenvalues and act in those directions. %

 Notice $(1-p^i_t) \leq 1$ is close to $1$ when $p^i_t$ is close to zero. Under normal conditions  $p^i_t \in [0,1]$ should be small.  
 We linearized the model around the origin and study the problem of regulating the state of the following system to the origin:
\begin{equation}
	\dot{p}^i_t =-\alpha_0 p^i_t +\bar \eta \frac1N \sum_{j=1}^Na_{ij}p^j_t + \beta_0 u^i_t, \quad 	t \in [0,T]
\end{equation}
where $u^i_t$ represents the control actions (via vaccinations or medications) at each node
and $\bar \eta = \eta N$.

 For finite networks this scaling by $\frac1N$ makes no difference since it can be absorbed into $a_{ij}$ or $\eta$. 
%
%
  Since the averaged strength of the connection  between any given group and the other groups in the network will adjust as new  groups  join the network, and since the model must retain the  relative strengths of the various groups, the infection strength $\eta$ is decreased at a rate $1/N$ as the model size increases. 



Let the quadratic cost associated with this problem be given by 
\begin{equation}
	\begin{aligned}
		J(u) =\sum_{i=1}^N \int_0^T  \big(q_t(p^i_t)^2 + (u^i_t)^2+(u^i_t-\frac1N\sum_{j=1}^Na_{ij}u^j_t)^2 \big)dt \\
	 + q_T(p^i_T)^2,
	\end{aligned} 
\end{equation}
 where $q_t, q_T \geq 0$.  
 We want to reduce the infection level, limit the actions cost and make sure that subpopulations on the network receive almost the same amount of resources as their weighted neighbors for equity purposes.  This forms a good example to illustrate the polynomial structures appearing in the cost functions
in \cite{ShuangAdityaCDC19,
ShuangPeterCDC19W2}. 

 The adjacency matrix $A=[a_{ij}]$ of an undirected contact network  has a spectral decomposition $A= \sum_{\ell=1}^L \mu_\ell v_\ell v_\ell^\TRANS $,  with $v_\ell$ representing the normalized eigenvector of the non-zero eigenvalue $\mu_\ell$. 
The number $L$  of non-zero eigenvalues can be much smaller than $N$.
The optimal solution \cite{ShuangAdityaCDC19} to be executed at community $i$ is given by 
$$
	u_t^i = \frac{\beta_0}{2} \breve \Pi_t p_t^i + \sum_{\ell=1}^L\Big(\frac{\beta_0\Pi_t^\ell}{(\frac{\mu_\ell}{N})^2-2\frac{\mu_\ell}{N}+2} -\frac{\beta_0\breve\Pi_t}{2}\Big)p_t^\TRANS v_\ell v_\ell(i), 
$$
where $\breve \Pi_t$ and $\Pi_t^\ell$ are given by
\begin{equation}\label{equ:riccati-finite}
	\begin{aligned}
		-\dot{\breve \Pi}_t &= -2\alpha_0\breve \Pi_t - \frac{\beta_0^2 (\breve \Pi_t)^2}{2} +q_t, \\
		-\dot{\Pi}^\ell_t &= -2(\alpha_0- \frac{\bar\eta \mu_\ell}{N})\Pi^\ell_t -\frac{\beta_0^2(\Pi_t^\ell)^2}{(\frac{\mu_\ell}{N})^2-2\frac{\mu_\ell}{N}+2}  +q_t, 
	\end{aligned}
\end{equation}
with $\breve  \Pi_T = \Pi^\ell_T=q_T$, 
and $p_t = [p_t^1,\ldots,p_t^N]^\TRANS$.

If the graphon limit $\FA$ exists,
then  the corresponding regulation problem for the graphon dynamical system, is given as follows: for a representative node $\gamma \in [\underline{\gamma}, \overline{\gamma}] \subset [0,1]$ on the graphon    
\begin{equation}
\begin{aligned}
		&\dot{\Fp}_t(\gamma) =-\alpha_0 \Fp_t(\gamma) +\bar \eta \int_0^1  \FA(\gamma, \rho) \Fp_t(\rho) d \rho+ \beta_0 \Fu_t(\gamma), \\
		&J(\Fu) = \int_0^T(\|\Fp_t\|^2_2+ \|\Fu_t\|_2^2 + \|(\BI -\FA)\Fu_t\|_2^2)dt + \|\Fp_T\|_2^2
\end{aligned}
\end{equation}
where $\Fp_t, \Fu_t \in L^2[0,1]$. Let $\FA = \sum_{\ell=1}^{\infty}\lambda_\ell  \Ff_\ell\Ff_\ell^\TRANS$.
Then the optimal solution \cite{ShuangPeterCDC19W2} is given by 
\begin{equation}
\begin{aligned}
	\Fu_t(\gamma) &= \frac{\beta_0}{2} \breve \Pi_t \Fp_t(\gamma) \\
	& \quad + \sum_{\ell=1}^{\infty}\big(\frac{\beta_0 \Pi_t^\ell}{2-2\lambda_\ell+ \lambda_\ell^2} - \frac{\beta_0}{2} \breve \Pi_t\big)\langle\Fp_t, \Ff_\ell\rangle \Ff_\ell(\gamma)
\end{aligned}
\end{equation}
where 
\begin{equation}\label{eq:riccati-infinite}
	\begin{aligned}
		-\dot{\breve \Pi}_t &= -2\alpha_0\breve \Pi_t - \frac{\beta_0^2 (\breve \Pi_t)^2}{2} +q_t, \\
		-\dot{\Pi}^\ell_t &= -2(\alpha_0- {\bar\eta \lambda_\ell})\Pi^\ell_t -\frac{\beta_0^2(\Pi_t^\ell)^2}{({\lambda_\ell})^2-2{\lambda_\ell}+2}  +q_t, 
	\end{aligned}
\end{equation}
with $\breve  \Pi_T = \Pi^\ell_T=q_T$.

Note that the only difference between \eqref{equ:riccati-finite} and \eqref{eq:riccati-infinite} lies in the eigenvalues $\mu_\ell/N$ and $\lambda_\ell$. This is consistent with the discussion on the eigenvalues of step function graphons following Proposition \ref{prop:spectral-decomposition}.

The global state aggregates (i.e. projections of states in different eigendirections) instead of local neighboring states are used in the local control of each subsystem. If $\FA$ is an approximation of the underlying network, then the corresponding approximate control can be applied. See \cite{ShuangPeterTAC18,ShuangPeterCDC19W2} for more related discussions.

\section{Numerical Illustration}
A numerical simulation is carried out for the controlled epidemic process with results shown in Fig. \ref{fig:sim}. Each node represents a city. The contact network among cities is represented by the air traffic frequencies among the corresponding city airports. The network data set USAir97 in \cite{USAir97} is used in the simulation, which represents a network of  332 American airports in 1997. Note that this example is only for the purpose of illustration, and many other network factors and population sizes should be included to have a better representation of coupling strength among  cities. 
%
The parameters for the numerical example are: $\alpha_0 = -0.5, \beta_0= 1, \eta = 1.5, q_t = 2, q_T = 4$, $T=1$  time unit.
In Fig. \ref{fig:sim}, the eigenstate and the eigencontrol in the $v_\ell$ direction (i.e. the projections of states and controls into the $v_\ell$ eigendirection) are given by $p_t^\TRANS  v_\ell v_\ell$ and $u_t^\TRANS v_l v_l$, respectively;
the auxiliary states and controls 
are given by $\breve p_t = p_t -\sum_{\ell=1}^L p_t^\TRANS v_\ell v_\ell$ and $\breve u_t = u_t -\sum_{\ell=1}^L u_t^\TRANS v_\ell v_\ell$, respectively.

\section{Conclusion}
 For controlling infinite dimensional systems and large-scale network systems, spectral properties are extremely useful in both control analysis and synthesis. 
Many important topics in the control of graphon dynamical systems still require further investigation. 
 First, systematic procedures for specifying graphon labellings and relabellings (in general, measure preserving transformations, see \cite{lovasz2012large}) need to be investigated. 
  Furthermore, it is of great interest to develop low-complexity control solutions to control large-scale networks with nonlinear local dynamics. 
We note that the study of controllability in this work is limited to the class of graphon systems where $\BA$ and $\BB$ share the same spectral structure. Hence further investigation is required to extend this study more general graphon dynamical systems. 
Finally, future investigations need to include the control analysis and synthesis of network systems coupled over exchangeable random graphs.

 \bibliographystyle{IEEEtran}
 \bibliography{IEEEabrv,mybib-abrv}

\end{document}